\documentclass[11pt]{article}
\usepackage[margin=1in]{geometry}                % See geometry.pdf to learn the layout options. There are lots.
\usepackage{graphicx}
\usepackage{amssymb}
\usepackage{amsmath}
\usepackage{amsthm}
\usepackage{microtype}
\usepackage{xcolor}
\usepackage{epstopdf}
 \usepackage{xr-hyper}
\usepackage{mathtools}
\usepackage[unicode=true,pdfusetitle,
 bookmarks=true,bookmarksnumbered=false,bookmarksopen=false,
 breaklinks=false,pdfborder={0 0 1},backref=false,colorlinks=true]
 {hyperref}
 \hypersetup{urlcolor=blue}
\hypersetup{citecolor=blue}
\hypersetup{filecolor=red}

\usepackage{setspace}
\makeatletter
\def\@setthanks{\vspace{-\baselineskip}\def\thanks##1{\@par##1\@addpunct.}\thankses}
\makeatother
\usepackage{natbib}
\usepackage{caption}
\usepackage{tabularx}
 \usepackage{float}
 \usepackage{booktabs}
 \usepackage{multirow}
\usepackage{bm}
\newtheorem{theorem}{Theorem}

\newtheorem{assum}{Assumption}[section]

\numberwithin{equation}{section}

\newtheorem{remark}{Remark}
\newtheorem{innercustomgeneric}{\customgenericname}
\providecommand{\customgenericname}{}
\newcommand{\newcustomtheorem}[2]{%
  \newenvironment{#1}[1]
  {%
   \renewcommand\customgenericname{#2}%
   \renewcommand\theinnercustomgeneric{##1}%
   \innercustomgeneric
  }
  {\endinnercustomgeneric}
}
\numberwithin{equation}{section}
\newcustomtheorem{customthm}{Theorem}
\newcustomtheorem{customlem}{Lemma}
\newcustomtheorem{customassumption}{Assumption}
\newcustomtheorem{customprop}{Proposition}
\newcustomtheorem{customexample}{Example}
\newcustomtheorem{customdef}{Definition}
\newcustomtheorem{customcor}{Corollary}
\newcustomtheorem{customrem}{Remark}

\usepackage[title]{appendix}

\begin{document}
\begin{titlepage}

\title{Group-Heterogeneous Changes-in-Changes and Distributional Synthetic Controls\thanks{We are grateful for helpful comments from Joshua Angrist and Simon Lee.}
}
\author{Songnian Chen\thanks{School of Economics, Zhejiang University. Email: \href{mailto:snchen2022@zju.edu.cn}{snchen2022@zju.edu.cn}.}\ \ \ \ \ \  Junlong Feng\thanks{Department of Economics, the Hong Kong University of Science and Technology. Email: \href{mailto:jlfeng@ust.hk}{jlfeng@ust.hk}.}}

\date{February 2026}                                          

\maketitle
\vspace{-3em}
\begin{abstract} 
We develop new changes-in-changes (CIC) and distributional synthetic controls (DSC) types of methods when
there exists group-level heterogeneity. For CIC, we allow individuals to belong to heterogeneous groups, extending \cite{athey2006identification} by finding appropriate control groups that share similar group-level unobserved characteristics to the treatment groups. For DSC, we show that the synthetic control units are not necessarily from the same period as in \cite{gunsilius2023distributional}; they may come from diﬀerent periods in which they have comparable group-level heterogeneity to the treatment group. Implementation of these new methods is briefly discussed.
\\
\vspace{0in}\\
\noindent\textbf{Keywords:} Causal inference, differences-in-differences, synthetic control, group heterogeneity.\\
\noindent\textbf{JEL Codes:} C21, C23, C31, C33\\
\bigskip
\end{abstract}
\end{titlepage}

\section{Introduction}\label{s1}

Individuals in economic data often belong to multiple heterogeneous groups. Group level heterogeneity is a common modeling feature in popular econometric methods, such as differences-in-differences and synthetic controls: In differences-in-differences, many empirical models include group-time fixed effects. In synthetic controls, models usually assume a group level outcome to be affected by latent common factors and group level factor loadings. 
However, despite the rich heterogeneity such group level unobservables capture, methods for nonseparable models, such as changes-in-changes method (CIC) in \cite{athey2006identification} and the distributional synthetic control (DSC) \citep{gunsilius2023distributional}, while focusing on individual level heterogeneity, often do not adequately control for group level unobservables.

In this paper, we introduce a model that allows for both individual level and group level unobservables, thereby extending the CIC and DSC frameworks to multivariate unobservables in the special case where one dimension operates only at the group level. As \cite{torous2024optimal} note, accommodating genuinely multivariate unobservables in a CIC-style setting is difficult; our approach provides a tractable and immediately applicable solution by restricting one factor to be common within groups while preserving the nonseparable structure of \cite{athey2006identification} and \cite{gunsilius2023distributional}. Under different data structures, we derive identification of heterogeneous quantile treatment effects on the treated by extending the methods in \cite{athey2006identification} and \cite{gunsilius2023distributional}. The changes-in-changes approach in \cite{athey2006identification} works with a small number of groups, whereas we are able to construct an appropriate comparison group to accommodate group level heterogeneity by working with a large number of groups. Similarly, for synthetic controls, \cite{gunsilius2023distributional} considers a setting with a small number of time periods,\footnote{The traditional synthetic control setting requires $T_{0}$, the number of pre-treatment periods, to be large.}  but we require a large number of both the pre- and post-treatment periods; we show that the appropriate synthetic control can be constructed to control for group level heterogeneity using possibly different time periods in which they share comparable group heterogeneity, in spirit, similar to the synthetic differences-in-differences approach in \cite{arkhangelsky2021synthetic}, in contrast to \cite{gunsilius2023distributional} where control groups at the same time period are used; we compare these approaches in more detail in Sections \ref{sec.comp.classical}-\ref{sec.comp.dsc}. In general, the presence of large numbers of groups and time periods is common in the differences-in-differences and synthetic control literature; a very incomplete list of examples include \cite{li2020statistical}, \cite{athey2021matrix}, \cite{callaway2021difference}, \cite{ferman2021properties} and \cite{arkhangelsky2024causal}. Our methods bring group-level unobservables into the nonparametric toolkit in a way that is both theoretically clean and empirically implementable with the data structures already common in the literature.

\subsection{The Model}
Suppose there are $G$ groups. In each group $g=1,\ldots,G$, there are $n_{g}$ individuals, denoted by $i=1,\ldots,n_{g}$. There are $T+1$ time periods and two sets of groups $\mathcal{G}^{I}$ and $\mathcal{G}^{N}$, where $\mathcal{G}^{I}\cap \mathcal{G}^{N}=\emptyset$ and $\mathcal{G}^{I}\cup \mathcal{G}^{N}=\{1,\ldots,G\}$. Treatment is assigned to groups in $\mathcal{G}^{I}$ in period $T_{0}+1$, whereas groups in $\mathcal{G}^{N}$ never receive the treatment. For an individual $i$ in group $g$ in time period $t$, her observable outcome and potential outcomes should she receive the treatment or not are denoted by $Y_{igt},Y_{igt}^{I}$ and $Y_{igt}^{N}$, respectively. Denoting the treatment assignment by $I_{igt}$, these outcomes satisfy
\begin{equation}\label{eq.potential}
  Y_{igt}=I_{igt}\cdot Y_{igt}^{I}+(1-I_{igt})\cdot Y_{igt}^{N}.
\end{equation}

Let $U_{igt}$ and $V_{gt}$ be two scalar random variables which econometricians do not observe. We assume that the potential outcome $Y_{igt}^{N}$ satisfies
\begin{equation}\label{eq.model}
  Y_{igt}^{N}=h(U_{igt},V_{gt},t),\forall i,g,t.
\end{equation}

The major difference in our model compared with those in \cite{athey2006identification} and \cite{gunsilius2023distributional} is the inclusion of $V_{gt}$. It allows the model to capture richer heterogeneity. For instance, consider peer effects of school integration \citep{angrist2004does,chetverikov2016iv}. Let $Y$ be the test score of a student. A group $g$ is a grade-school cell. Let $U_{igt}$ be the unobservable learning ability of student $i$ in group $g$ and year $t$. Let $V_{gt}$ be a scalar index of group heterogeneity that captures the shared unobserved characteristics such as teacher quality and classroom environment across grades, schools and years \citep{krueger1999experimental} as a scalar latent factor. If one did not include $V_{gt}$ in the model, the model would imply that, for instance, as long as students of the same grade in the same year but in two different schools have the same level of learning ability, then even if the teachers' quality and classroom environment in these schools were very different, these students would have achieved the same test scores. In contrast, our model avoids this empirically implausible implication by allowing variation in $Y^{N}_{igt}$ due to group heterogeneity. The separable structure captures settings where group-level factors such as teacher quality affect all individuals within a group through a common channel, rather than operating
only through selection on individual unobservables.

With the presence of $V_{gt}$, constructing the counterfactual distribution for units in the treatment group needs to account for group heterogeneity. Specifically, for some $\tau_{U}\in (0,1)$, the $\tau_{U}$-th quantile of the distribution of $Y^{N}_{igt}$ for a fixed group $g$ and time period $t$ depends on (the realization of) $V_{gt}$. Hence, one has to find units or groups in the control group who share comparable level of $V_{gt}$. Ignoring such heterogeneity generally leads to incorrect construction of the counterfactuals.

Our key methodological contribution is to show how large cross-sections of groups (in CIC) or long time series (in DSC) allow us to match explicitly on group-level unobservables $V_{gt}$. This yields new, practical ways to construct counterfactual distributions in the empirically common setting of clustered data with both individual and group heterogeneity.

In what follows, we define the counterfactuals we are interested in and discuss its identification in the settings of differences-in-differences and synthetic controls. We then briefly sketch implementation and data requirements. We collect all the proofs in the Appendix.

\section{Changes-in-Changes with Group Heterogeneity}\label{sec.did}
We first consider a typical two-period differences-in-differences or changes-in-changes setup where $T_{0}=0$ and $T=1$. For a fixed group $g$, a fixed time period $t$, and $k\in\{I,N\}$, let $Y_{gt}^{k}(\tau_{U})$ be the $\tau_{U}$-th quantile of $Y^{k}_{igt}$ within $g$ and $t$. Note that $Y_{gt}^{k}(\tau_{U})$ is a random variable whose randomness solely comes from $V_{gt}$. Let the $\tau_{V}$-th quantile of $Y^{N}_{gt}(\tau_{U})$ over $g$ for a fixed $t$ be $Y^{N}_{It}(\tau_{U},\tau_{V})$ when $g\in\mathcal{G}^{I}$ and by $Y^{N}_{Nt}(\tau_{U},\tau_{V})$ when $g\in\mathcal{G}^{N}$. In the test score example, $Y_{gt}^{N}(\tau_{U},\tau_{V})$ is the potential test score for a particular student whose unobservable learning ability is ranked at $\tau_{U}$ within her/his group, and whose group possesses teacher quality and classroom environment ranked at $\tau_{V}$ among all grade-school groups. By equation \eqref{eq.potential}, $Y^{N}_{kt}(\tau_{U},\tau_{V})$ can be directly identified using the distributions of $Y_{igt}$ for all $k$ and $t$ except for $(k,t)=(I,1)$. In this section, we propose assumptions to identify $Y^{N}_{I1}(\tau_{U},\tau_{V})$.

\begin{assum}\label{assum.did.mono}
The production function $h(u,v,t):S(U)\times S(V)\times\{0,1\}\mapsto \mathbb{R}$ is componentwise strictly increasing in $u$ and $v$ for $t\in\{0,1\}$. 
\end{assum}

\begin{assum}\label{assum.did.ind}
$U_{igt}\perp V_{gt}$ for all $i,g,t$.
\end{assum}

\begin{assum}\label{assum.did.iid}
$U_{igt}$ and $V_{gt}$ are both identically distributed in $i,g,t$ for $i=1,\ldots,N_{g}$, $t=0,1$ and $g\in\mathcal{G}^{I}$ and $g\in\mathcal{G}^{N}$, with quantile functions denoted by $Q_{U_{I}},Q_{U_{N}},Q_{V_{I}}$ and $Q_{V_{N}}$, respectively. 
\end{assum}

Finally, let $S(U_{I}),S(V_{I})$, $S(U_{N})$ and $S(V_{N})$ denote the support of $U_{igt}$ and $V_{gt}$ for $g$ in $\mathcal{G}^{I}$ and $\mathcal{G}^{N}$, respectively. 
\begin{assum}\label{assum.did.support}
$S(U_{I})\subseteq S(U_{N})$ and $S(V_{I})\subseteq S(V_{N})$. 
\end{assum}

Similar to \cite{athey2006identification}, Assumption \ref{assum.did.mono} assumes $h$ to be strictly increasing in the individual heterogeneity. Besides, we also assume strict monotonicity in the scalar group level unobservable. We need these two assumptions to link $h(\cdot,\cdot,t)$ to the cross-group quantile of the within-group quantile of $Y$. These monotonicity assumptions are natural when, for example, $U_{igt}$ and $V_{gt}$ capture individual and group level characteristics that can be ranked, such as learning ability, teacher quality and classroom environment. For example, in the schooling application, a higher learning ability or better teacher quality/classroom environment raises test scores for any fixed level of the other unobservable, which is the economic content of strict monotonicity in both arguments. In particular, for $V_{gt}$, while teacher quality and classroom environment could in principle trade off within a group, the scalar $V_{gt}$ represents their composite effect on the production function, which is the economically relevant object for ranking groups.

 Assumption \ref{assum.did.ind} assumes independence between the two unobservables. It holds if group assignment/sorting is driven by individual unobservables $U_{igt}$ but not by group-level unobservables $V_{gt}$. This is analogous to the standard exogeneity of fixed effects in nonlinear panel models \citep{bai2009panel,moon2015linear}. In the schooling example, the independence assumption rules out perfect sorting of high-ability students into schools with unobservably better teachers; students may sort on observable school quality or their own ability, but teacher quality is assigned largely exogenously (or is hard to observe ex ante). The strength of this assumption is that it separates the quantiles of $U$ and $V$. If sorting on $V_{gt}$ were important, one would need additional instruments or a different modeling approach.

Assumption \ref{assum.did.iid} requires only that the pools of treated and control groups have the same unobservable distributions (e.g., student ability distributions are balanced across treated vs. untreated schools, and teacher-quality distributions are stable). This is weaker than requiring identical groups. Its strength is that it delivers clean quantile matching with large $G$; its limitation is that it could be violated if treatment itself changes group composition. Note that, although Assumption \ref{assum.did.iid} restricts the distributions of the unobservables to be stable over time, the distribution of outcomes can still change across periods because the production function $h(\cdot,\cdot,t)$ itself is allowed to vary with t. 

Together with the restrictions in Theorem \ref{thm.did.id}, Assumptions \ref{assum.did.ind} and \ref{assum.did.iid} are our version of the \textit{parallel trend} condition in the differences-in-differences literature. They imply stable counterfactual distributions within treatment/control groups over time. Note that this naturally holds under the assumptions in \cite{athey2006identification}, where $V_{gt}$ is assumed away and $U_{igt}=U_{ig't}$ for all $g,g'\in\mathcal{G}^{I}$ or $\mathcal{G}^{N}$. Hence, our assumptions do not impose extra restrictions than that classical case under group homogeneity.

\begin{remark}
We can relax Assumption \ref{assum.did.iid} to allow for certain time series heteroscedasticity by allowing the distribution functions of $U_{ig1}$ and $V_{g1}$ to be strictly monotone transformations of $U_{ig0}$ and $V_{g0}$, respectively, as long as such transformations do not vary in $i$ and $g$. This is because these transformations can be without loss of generality absorbed into the last argument of $h$. 
\end{remark}

For fixed $g$ and $t$, let $Y_{gt}(\tau_{U})$ be the $\tau_{U}$-th quantile of $Y_{igt}$ over $i$. This quantity is an observable random variable and equal to $Y_{gt}^{N}(\tau_{U})$ for all $g$ when $t=0$ and for $g\in\mathcal{G}^{N}$ when $t=1$. For such $g$ and $t$, denote the group quantile function and the cumulative distribution function (CDF) of $Y_{gt}(\tau_{U})$ by $Q_{Y_{It}(\tau_{U})}(\cdot)$ and $F_{Y_{It}(\tau_{U})}(\cdot)$ if $g\in\mathcal{G}^{I}$ and by $Q_{Y_{Nt}(\tau_{U})}(\cdot)$ and $F_{Y_{It}(\tau_{U})}(\cdot)$ if $g\in\mathcal{G}^{N}$, respectively; these functions are directly identified in the population. We have the following result.

\begin{theorem}\label{thm.did.id}
Under Assumptions \ref{assum.did.mono}-\ref{assum.did.support}, for any fixed $(\tau_{U}^{*},\tau_{V}^{*})\in (0,1)\times (0,1)$, the following statements are true:
\begin{enumerate}
  \item[(i)] If $Q_{U_{I}}=Q_{U_{N}}$, then $Y_{I1}^{N}(\tau_{U}^{*},\tau_{V}^{*})=Q_{Y_{N1}(\tau_{U}^{*})}(F_{Y_{N0}(\tau_{U})}(Q_{Y_{I0}(\tau_{U})}(\tau_{V}^{*})))$ for any $\tau_{U}\in (0,1)$. 
  \item[(ii)] If $Q_{V_{I}}=Q_{V_{N}}$, then $Y_{I1}^{N}(\tau_{U}^{*},\tau_{V}^{*})=Q_{Y_{N1}(\tau_{U}')}(\tau_{V}^{*})$, where $\tau_{U}'$ satisfies $Q_{Y_{N0}(\tau_{U}')}(\tau_{V})=Q_{Y_{I0}(\tau_{U}^{*})}(\tau_{V})$ for any $\tau_{V}\in (0,1)$.  
  \item[(iii)] If there exists a coordinatewise strictly increasing function $\gamma:S(U)\times S(V)\mapsto \mathbb{R}$ such that $Y^{N}_{igt}=h(\gamma(U_{igt},V_{gt}),t)$ where $h(\cdot,t)$ is strictly increasing for $t=0,1$, then
$Y_{I1}^{N}(\tau_{U}^{*},\tau_{V}^{*})=Q_{Y_{N1}(\tau_{U}')}(\tau_{V}')$, where $(\tau_{U}',\tau_{V}')\in \{(\tau_{U},\tau_{V})\in (0,1)^{2}:Q_{Y_{N0}(\tau_{U})}(\tau_{V})=Q_{Y_{I0}(\tau_{U}^{*})}(\tau_{V}^{*})\}$ and the set is nonempty. 
\end{enumerate}
\end{theorem}

Theorem \ref{thm.did.id} establishes identification of $Y^{N}_{I1}(\tau_{U}^{*},\tau_{V}^{*})$ in three scenarios. Together with Assumption \ref{assum.did.iid}, case (i) requires that the distributions of $U_{igt}$ in the treatment and control groups are identical, but allows $V_{gt}$ to have a different distribution for $g\in\mathcal{G}^{I}$ from that for $g\in\mathcal{G}^{N}$. In this case, distributional heterogeneity in the unobservables is fully absorbed by the group level unobservable $V_{gt}$. This can be the case when, for instance, the learning ability ($U_{igt}$) has the identical distributions in all groups ($g$), but the treated groups in $t=1$ ($\mathcal{G}^{I}$) have different distributions from those untreated groups ($\mathcal{G}^{N}$). 

In contrast, the condition in case (ii) allows for distributional differences in $U_{igt}$ across the treatment and control groups. In this case, the requirements on $U_{igt}$ specified in Assumptions \ref{assum.did.iid} and \ref{assum.did.support} become identical to Assumptions 3.3 and 3.4 in \cite{athey2006identification}, which is indeed a special case of our setup because their $V_{gt}$ is a constant over $g$ and $t$ so our condition in case (ii) is directly satisfied. 

For both cases (i) and (ii), even when $G = 2$ as in \cite{athey2006identification}, in general, $v_{1t}\neq v_{0t}$ for realized values of the group level heterogeneity, which implies that the treatment and control groups would have different production functions. With a large $G$, we essentially are applying the \cite{athey2006identification} approach to the control and treatment groups with $V_{gt}$ at the same quantile level; in particular, for case (ii), $V_{gt}$ is controlled at the same value by the identical distribution assumption for $V_{gt}$.

Case (iii) allows for distributional differences across the treatment and control groups in both $U_{igt}$ and $V_{gt}$. As a cost, an extra single index structure is needed for the production function $h$. 
This restriction, though not fully general, is economically natural whenever individual and group unobservables enter outcomes through a common latent channel (e.g., “effective student quality” as a monotonic function of ability and school resources). It preserves the nonseparable flavor of the model while allowing monotone-rearrangement arguments \citep{matzkin2003nonparametric}. The fully nonparametric case without any index structure remains an important open challenge that would likely require entirely new identification arguments; we view the present results as a tractable and immediately usable intermediate step.

Theorem \ref{thm.did.id} has a clear interpretation about how one should construct the counterfactuals for the treatment group in period $1$. Due to the differences in the distributions of $U_{igt}$ or $V_{gt}$ in the treatment and control groups, $U_{igt}$ and/or ${V}_{gt}$ for the treatment group at the given quantiles $(\tau_{U}^{*},\tau_{V}^{*})$ may rank differently in the control group. So, we need to find an appropriate ``comparison group'' by matching the unobservables at the correct quantiles in $\mathcal{G}^{N}$. Once the quantiles are matched using the variation in period 0, the same quantiles are still matched in period 1 since the distributions of $U_{igt}$ and $V_{gt}$ do not change over time, yielding our identification result.

\subsection{Testable Implications}

Theorems \ref{thm.did.id} provides three identification equations under different assumptions on the distributions of $U_{igt}$ and $V_{gt}$. Since the expressions of $Y^{N}_{I1}(\tau_{U}^{*},\tau_{V}^{*})$ in the three cases are different, it is useful to know in practice which world one lives in. The following theorem provides testable implications under an extra assumption on the support of the unobservables.
\begin{assum}\label{assum.did.compact}
The support sets $S(U_{I})$ and $S(V_{I})$ are compact.
\end{assum}
We have the following theorem.
\begin{theorem}\label{thm.did.test}
Under Assumptions \ref{assum.did.mono}-\ref{assum.did.compact}, the following statements are true:
\begin{enumerate}
  \item[(i)] The condition in Theorem \ref{thm.did.id}-(i) holds if and only if for each $\tau_{V}\in[0,1]$, there exists a $\tau_{V}'$ which does not depend on $\tau_{U}$ such that $Q_{Y_{N0}(\tau_{U})}(\tau_{V}')=Q_{Y_{I0}(\tau_{U})}(\tau_{V})$.
  \item[(ii)] The condition in Theorem \ref{thm.did.id}-(ii) holds if and only if for each $\tau_{U}\in [0,1]$, there exists a $\tau_{U}'$ which does not depend on $\tau_{V}$ such that $Q_{Y_{N0}(\tau_{U}')}(\tau_{V})=Q_{Y_{I0}(\tau_{U})}(\tau_{V})$.
  \item[(iii)] The result in Theorem \ref{thm.did.id}-(iii) holds for all $(\tau_{U}^{*},\tau_{V}^{*})\in [0,1]^{2}$ if for all $(\tau_{U},\tau_{V})$ and $(\tau_{U}',\tau_{V}')$ in $[0,1]^{2}$, $Q_{Y_{N1}(\tau_{U})}(\tau_{V})=Q_{Y_{N1}(\tau_{U}')}(\tau_{V}')$ holds whenever $Q_{Y_{N0}(\tau_{U})}(\tau_{V})=Q_{Y_{N0}(\tau_{U}')}(\tau_{V}')$. Suppose Assumption \ref{assum.did.support} is strengthened as $S(U_{I})=S(U_{N})$ and $S(V_{I})=S(V_{N})$, the converse is also true.
\end{enumerate}
\end{theorem}

It is worth noting that only compactness of $S(U_{I})$ in Assumption \ref{assum.did.compact} is needed to show Theorem \ref{thm.did.test}-(i). In combination with Assumption \ref{assum.did.support}, it guarantees that $Q_{U_{I}}(\cdot)$ and $Q_{U_{N}}(\cdot)$ intersect at least once on $[0,1]$. The quantile crossing requirement for $U_{igt}$ for $g\in\mathcal{G}^{I}$ and $\mathcal{G}^{N}$ is mild because here we do not need any additional restrictions on $V_{gt}$, which mainly accounts for the systematic in this case. Compactness of $S(V_{I})$ is needed for a similar reason to show Theorem \ref{thm.did.test}-(ii). 

Theorem \ref{thm.did.test} says that one can in principle check which of the three identification results in Theorem \ref{thm.did.id} to use. Since all three conditions in Theorem \ref{thm.did.id} have if-and-only-if testable implications, it is conclusive to use those results when the corresponding implications hold. Although stated in population terms, these implications are assessable empirically. One can replace the population quantile functions with uniformly consistent sample estimators (as detailed in Section \ref{sec.est.cic}) and test whether the relevant integrated squared differences\footnote{Using case (ii) as an example, one can verify whether the sample analogue of $\int_{0}^{1}(\min_{\tau_{U}'\in [0,1]}\int_{0}^{1} (Q_{Y_{N0}(\tau'_{U})}(\tau_{V})-Q_{Y_{I0}(\tau_{U})}(\tau_{V}))^{2}d\tau_{V})d\tau_{U}$ is equal to 0. } (or the equality of quantile surfaces) are statistically close to zero using bootstrap critical values or asymptotic Kolmogorov–Smirnov-type statistics for quantile processes \citep{chernozhukov2013inference}. Such tests would help practitioners choose the appropriate case of Theorem \ref{thm.did.test} and provide a natural starting point for formal specification testing in future work.

 Note that the three conditions in Theorem \ref{thm.did.id} are not mutually exclusive, so it is possible that multiple results in Theorem \ref{thm.did.id} hold simultaneously, yielding overidentification.

\section{Distributional Synthetic Control with Group Heterogeneity}\label{sec.sc}
Unlike the differences-in-differences setup where $T_{0}=0$ and $T=1$, we adopt the classical synthetic control setting where $G$ is fixed but $T_{0}$ is large. Treatment is only given to group $g=1$ in period $T_{0}+1$. Meanwhile, we assume $T-T_{0}$ is large as well. 

Our goal is to construct a synthetic control for the post-treatment period counterfactual $Y_{1t}^{N}(\tau_{U},\tau_{V})$ for $t>T_{0}$ for arbitrary $(\tau_{U},\tau_{V})$, which, recalling the definition of this notation, refers to the $\tau_{V}$-th time series quantile in period $t$ of the $\tau_{U}$-th individual quantile of $Y^{N}_{i1t}$. We now introduce assumptions such that $Y_{1t}^{N}(\tau_{U},\tau_{V})$ is unchanged over $t$ for all $t>T_{0}$ or $t\leq T_{0}$, denoted by $Y_{1,post}^{N}(\tau_{U},\tau_{V})$, and $Y_{1,pre}^{N}(\tau_{U},\tau_{V})$ respectively.
 
\begin{assum}\label{assum.sc.mono}
The production function $h(u,v,t)=h(u,v,post)$ for all $t>T_{0}$, and $h(u,v,t)=h(u,v,pre)$ for all $t\leq T_{0}$.
\end{assum}
We relax this assumption in Section \ref{sec.trend} below to allow additive (group-homogeneous or group-heterogeneous) time trends, so that the production function can change smoothly over time while still being constant within the pre- and post-treatment windows after de-trending.

\begin{assum}\label{assum.sc.iid}
$U_{igt}$ has identical distribution in $i$, $g$ and $t$ for $t\leq T_{0}$ and $t>T_{0}$, with quantile functions $Q_{U_{pre}}$ and $Q_{U_{post}}$, respectively. $V_{gt}$ has identical distribution in $t$ for $t\leq T_{0}$ and $t>T_{0}$, with quantile functions $Q_{V_{g,pre}}$ and $Q_{V_{g,post}}$, respectively.
\end{assum}

Similar to \cite{gunsilius2023distributional}, we do not impose any monotonicity condition on $h$. Under the independence condition Assumption \ref{assum.did.ind}, Assumption \ref{assum.sc.iid} implies that for any $\tau_{U}\in (0,1)$, the $\tau_{U}$-quantile of $Y_{igt}^{N}$ in group $g$ and period $t$ satisfies
\begin{equation*}\label{eq.sc.represent}
  Y_{gt}^{N}(\tau_{U})=\begin{cases}\tilde{h}(V_{gt},pre,\tau_{U})&t\leq T_{0},
  \\\tilde{h}(V_{gt},post,\tau_{U}),&t> T_{0}, \end{cases}
\end{equation*}
for some $\tilde{h}$ function. 
\begin{remark}
Our assumptions rule out time trends in the outcome variables; for all $t\leq T_{0}$ or $t>T_{0}$, the group-wise or unconditional mean of the outcome is constant in time since $h$ stays unchanged and $(U_{igt},V_{gt})$ are identically distributed in $t$. We extend our model in Section \ref{sec.trend} to include certain time trends.
\end{remark}

Due to the heterogeneity in the realization of $V_{gt}$, the approach in \cite{gunsilius2023distributional} by directly using $Y_{gt}(\tau_{U})$ ($g>1$) at the same time period $t$ to construct the synthetic control may not be appropriate. Instead, for each $g>1$, we reshuffle $Y_{gt}(\tau_{U})$ over time to construct the synthetic control so that the $V_{gt_{g}}$s ($g>1$ and $t_{g}$ may be different for different $g$) are comparable to the targeted $V_{1t}$. This idea is similar to the synthetic differences-in-differences approach in \cite{arkhangelsky2021synthetic}. We will discuss the relationship between our approach and these alternatives in more detail in Sections \ref{sec.comp.sdid} and \ref{sec.comp.dsc}.

Specifically, let $Q_{Y_{g,pre}(\tau_{U})}(\tau_{V})$ and $Q_{Y_{g,post}(\tau_{U})}(\tau_{V})$ denote the quantile function of $Y_{gt}(\tau_{U})$ for $t\leq T_{0}$ and $t>T_{0}$, respectively. Assume weights $\lambda^{*}(\tau_{U})\coloneqq \{\lambda_{g}^{*}(\tau_{U})\}_{g=2,\ldots,G}\in \Delta^{G-1}$ exist, where $\Delta^{G-1}$ denotes the $(G-1)$ dimensional simplex, such that 
\begin{equation}\label{eq.sc.weights}
  Q_{Y_{1,pre}(\tau_{U})}(\tau_{V})=\sum_{g=2}^{G}\lambda_{g}^{*}(\tau_{U})Q_{Y_{g,pre}(\tau_{U})}(\tau_{V}),\ \forall \tau_{V}\in (0,1).
\end{equation}
We can use these weights to construct the synthetic control under an isometry condition on $\tilde{h}$.
\begin{theorem}\label{thm.sc.id}
Let Assumptions \ref{assum.did.ind}, \ref{assum.sc.mono} and \ref{assum.sc.iid} hold. If $\tilde{h}(\cdot,j,\tau_{U})$ is a scaled isometry on the 2-Wasserstein space for all $j\in\{pre,post\}$, and if the maps $V_{gt}\mapsto V_{gt'}$ for all $g=1,\ldots,G$ are such that they preserve the relative weights $\lambda^{*}(\tau_{U})$ between the probability measures $P_{V_{1t}}$ and $P_{V_{gt}}$ for each $g>1$, then $Y_{1,post}^{N}(\tau_{U},\tau_{V})=\sum_{g=2}^{G}\lambda_{g}^{*}(\tau_{U})Q_{Y_{g,post}(\tau_{U})}(\tau_{V})$ for all $\tau_{V}$.
\end{theorem}
\begin{remark}
If the isometry condition is strengthened to hold for all $\tau_{U} \in (0,1)$, the optimal weights become independent of $\tau_{U}$, recovering the standard constant-weight synthetic control (including \cite{gunsilius2023distributional}, under homogeneity). The more general $\tau_{U}$-dependent weights are, however, a central feature of our framework and are empirically relevant whenever group-level heterogeneity interacts with individual unobservables. For example, high-ability students (high $\tau_{U}$) may benefit disproportionately from high teacher quality (high $V_{gt}$), so that the “best” synthetic control groups differ across the outcome distribution. Allowing the optimal weights to vary therefore captures rich, quantile-specific group matching that constant-weight methods cannot accommodate and directly addresses settings with heterogeneous group effects—precisely the motivation for our group-heterogeneous DSC.
\end{remark}

Our synthetic control, as the counterfactual quantile function of $Y_{1t}^{N}(\tau_{U})$ for $t>T_{0}$, is constructed using the time series quantile functions of $Y_{gt}(\tau_{U})$ for $g>1$. To identify these quantile functions, we need large $T_{0}$ and $T-T_{0}$ so that the time series distributions of the $Y_{gt}(\tau_{U})$s are identified for each $g$ for $t\leq T_{0}$ and $t>T_{0}$. The requirement of a large $T_{0}$ is similar to the classical synthetic control methods. 

In contrast to \cite{gunsilius2023distributional}, our synthetic control is not directly using the $Y_{gt}(\tau_{U})$s ($g>1$) at the same period $t$. We use the $Y_{gt_{g}}(\tau_{U})$s ($g>1$) such that the $t_{g}$s may be different but make the $Y_{gt_{g}}(\tau_{U})$s be at the same quantile level in their time series distributions. 

In the following subsections, we first discuss how to generalize our model to include a time trend. We will then compare our methods with related approaches in detail.

\subsection{Time Trend}\label{sec.trend}
Our model can be extended to include an additive nonparametric group-homogeneous time trend $\text{tr}(t)$ or an additive parametric group-heterogeneous trend $\text{tr}(t;\theta_{g})$. Thus, the assumption that $h$ is constant within the pre- and post-treatment regimes is not as restrictive as it first appears; we now show that any smooth time trend can be removed nonparametrically or parametrically before applying the synthetic control.

We first consider a homogeneous time trend $\text{tr}(t)$ where $\text{tr}(0)$ is normalized to be 0. Now our model is $Y_{igt}=h(U_{igt},V_{gt},t)+\text{tr}(t)$ for $g\in \mathcal{G}^{N}$ for all $t$. Under Assumption \ref{assum.sc.iid}, trend $\text{tr}(t)$ is identified for every $t$ as the difference of the means of $Y_{igt}$ over the control groups between period $t$ and period $0$. 

For an additive group-specific time trend $\text{tr}(t;\theta_{g})$ where $\text{tr}(\cdot;\theta_{g})$ is a known function up to a finite-dimensional parameter $\theta_{g}$, we can identify $\theta_{g}$ for every $g$ by solving moment equations formed by $\partial_{t}\text{tr}(t;\theta_{g})=\partial_{t}\mathbb{E}(Y_{igt})$ for each $g$ and $t\leq T_{0}$, where the expectation is taken over $i$ for fixed $t$ and $g$. 

Once the trend is identified, all the previous analysis follows by subtracting the trend from $Y_{igt}$. 

\subsection{Comparison with the Classical Synthetic Control Methods}\label{sec.comp.classical}
\begin{sloppypar}The classical synthetic control methods are usually applied to aggregate level data \citep{abadie2003economic,abadie2010synthetic,abadie2015comparative,abadie2021using}. For individual level data, one first performs certain aggregation to obtain group level data, then constructs the synthetic control at the group level \citep{abadie2021using}. Our method naturally mimics this procedure: The construction of $Y_{gt}(\tau_{U}^{*})$, i.e., the within group quantile of $Y_{igt}$, is a nonparametric analogue of the aggregation step. Then the construction of the group level distributional synthetic control is analogous to the classical synthetic control methods for group level data. For the second step to work, both approaches require a large $T_{0}$, whereas our approach also needs a large $T-T_{0}$.\end{sloppypar}

\begin{table}[htbp]
\centering
  \caption{Comparison with Classical Methods}
  \label{tab:comp1}%
    \begin{tabular}{@{}ccc@{}@{}}
    \hline
    \hline
    & Our Method & Classical Methods \\
    \hline
    Step 1 (``Aggregation'') & \textit{Within group quantiles}  & \textit{Within group average/sum}  \\
    \hline
    \hspace{-7.8em} Step 2 & \multicolumn{2}{c}{\textit{Construct synthetic control with group level data}}\\
    \hline
    \hspace{-5.4em}Time series &  \textit{Large $T_{0}$ and $T-T_{0}$} & \textit{Large $T_{0}$}  \\
    \hline
    \end{tabular}%
  
\end{table}%

From a modeling perspective, our model with both individual and group heterogeneity is close to the classical methods. Adapting a linear factor model \citep{abadie2010synthetic,abadie2021using} without covariates to the individual level, we have
\begin{equation}\label{eq.factor}
Y_{igt}=  \delta_{t}+\mu_{g}'\theta_{t}+U_{igt},
\end{equation}
where $\delta_{t}$ is a common trend, $\mu_{g}$ and $\theta_{t}$ are unobserved factor loadings and common factors. Our model satisfies equation \eqref{eq.factor} by letting $V_{gt}\coloneqq \delta_{t}+\mu_{g}'\theta_{t}$. Taking the group level average and denoting $\bar{Y}_{gt}\coloneqq \sum_{i}^{n_{g}}Y_{igt}/n_{g}$ and $\varepsilon_{gt}\coloneqq \sum_{i}^{n_{g}}U_{igt}/n_{g}$, we then have exactly the linear factor model in \cite{abadie2010synthetic} and \cite{abadie2021using}:
\begin{equation*}\label{eq.factor.agg}
  \bar{Y}_{gt}=\delta_{t}+\mu_{g}'\theta_{t}+\varepsilon_{gt}.
\end{equation*}
Hence, our $V_{gt}$ captures the factor structure in the classical synthetic control methods. 
\subsection{Comparison with Synthetic Differences-in-Differences (SDID)}\label{sec.comp.sdid}
\cite{arkhangelsky2021synthetic} proposes a method that combines features of both the classical synthetic control method and the differences-in-differences methods. Similar to the classical synthetic control method, they also take a factor model as the data generating process. In contrast to the synthetic control method which constructs the ``control'' only using weights that are unchanged over time, SDID also estimates time weights so that different pre-treatment periods can have different weights when forming the control. 

Although our model and thus our method are different from theirs, this feature of their method echoes the spirit of ours; the treatment groups and control groups at the same time period may not necessarily be comparable. We take care of such possible incompatibility by reshuffling the time series and matching the time series quantiles, whereas they adjust the relevance of different time periods by weights.

\subsection{Comparison with \cite{gunsilius2023distributional}}\label{sec.comp.dsc}
Group level unobservable $V_{gt}$ is not present in \cite{gunsilius2023distributional}, where $Y^{N}_{igt}$ is assumed to be equal to $h(U_{igt},t)$. For any realization of $U_{igt}=u$, the model implies $Y^{N}_{igt}=h(u,t)$, which by construction can only change in time, but is fixed across all groups. Our model, on the other hand, allows for both time and group heterogeneity because $V_{gt}$'s realization changes in $g$ and $t$. Moreover, we can also allow deterministic time trend as discussed in Section \ref{sec.trend}.

When group level heterogeneity indeed exists, the synthetic control in \cite{gunsilius2023distributional} may not fit the counterfactual well. Suppose our model is true. Denote the realized $Y_{gt}(\tau_{U})$ in each $g$ and $t$ by $y_{gt}(\tau_{U})$. \cite{gunsilius2023distributional} obtains weights $\lambda_{g}$s by the following equation under our notation:
\begin{equation}\label{eq.sc.weightsG}
  y_{1t}(\tau_{U})=\sum_{g=2}^{G}\lambda_{g}y_{gt}(\tau_{U}),t\leq T_{0},\ \forall \tau_{U}\in (0,1).
\end{equation}
Note that $y_{1t}(\tau_{U})=\tilde{h}(v_{gt},pre,\tau_{U})$ where $v_{gt}$ is the realization of $V_{gt}$. Unless $V_{gt}$ does not depend on $g$ so that $v_{gt}$ is identical for all $g=1,\ldots,G$ in period $t$, the quantile levels of the time series distributions of the $Y_{gt}(\tau_{U})$s that these $y_{gt}(\tau_{U})$s correspond to are, in general, different. For concreteness, let $\tau_{V,g}$ be such that $Q_{Y_{g,pre}(\tau_{U})}(\tau_{V,g})=y_{gt}(\tau_{U})$. Then equation \eqref{eq.sc.weightsG} for one fixed $t$ is equivalent to the following equation:
\begin{equation*}
  Q_{Y_{1,pre}(\tau_{U})}(\tau_{V,1})=\sum_{g=2}^{G}\lambda_{g}Q_{Y_{g,pre}(\tau_{U})}(\tau_{V,g}),
\end{equation*}
 where the $\tau_{V,g}$s ($g=1,\ldots,G$) are unknown and in general not equal to each other. Therefore, weights $(\lambda_{g})$ in \cite{gunsilius2023distributional} are not obtained by matching the quantile functions of $Y_{gt}(\tau_{U})$; they are from matching the quantile function values at \textit{different quantile levels}. Consequently, 
 such weights do not in general yield the barycenter of the 2-Wasserstein space of the distributions of $Y_{g,pre}(\tau_{U})$s. The counterfactuals in $t>T_{0}$ constructed using these weights are not necessarily valid. 
 \begin{table}
 \centering
  \caption{Comparison with \cite{gunsilius2023distributional}}
  \label{tab:comp2}%
    \begin{tabular}{@{}ccc@{}@{}}
    \hline
    \hline
    & Our Method & \cite{gunsilius2023distributional} \\
     \hline
    Weights & \textit{Same time series quantile level} &\textit{Same time periods}\\
     \hline
    Time periods  & \textit{Long}  & \textit{Short}  \\
    \hline
    \end{tabular}%
  \end{table}%

In contrast, our method does not restrict the construction of the synthetic control to the same period. Instead, we utilize time series variation to guarantee that the weights are obtained by matching the entire quantile function $Q_{Y_{1,pre}(\tau_{U})}(\cdot)$; see equation \eqref{eq.sc.weights}. Under the isometry conditions in Theorem \ref{thm.sc.id}, weights obtained in this way in the pre-treatment periods are valid in post-treatment periods. 

As for the different data requirements, we emphasize that such difference is due to the different goals. \cite{gunsilius2023distributional} identifies the counterfactual using only two periods and cross-sectional variation across groups at the same time; our method uses long time series for each group to identify the full time-series quantile functions of the within-group quantiles and thereby match on $V_{gt}$. This is the natural analogue of the classical synthetic control literature, which already requires large $T_{0}$. In short, different data requirements reflect the different identifying sources.

\section{Implementation}
We sketch estimation and the data requirements in this section.
\subsection{Changes-in-Changes}\label{sec.est.cic}
For simplicity, we only consider estimating $Y^{N}_{I1}(\tau_{U}^{*},\tau_{V}^{*})$ under condition (i) in Theorem \ref{thm.did.id}; the other two cases are similar. All the assumptions in Section \ref{sec.did} hold. 

Suppose we have a data set $\{Y_{igt}:i=1,\ldots,n_{g};g=1,\ldots,G;t=0,1\}$. For a fixed $t=0,1$, assume $Y_{igt}$ is independent of $Y_{i'g't}$ for $(i,g)\neq (i',g')$ where $g,g'$ are both in $\mathcal{G}^{I}$ or $\mathcal{G}^{N}$. For each $g,t$ and an arbitrary $\tau_{U}\in (0,1)$, we can estimate $Y_{gt}(\tau_{U})$ by the $\tau_{U}$-th sample quantile of $Y_{igt}$, denoted by $\hat{Y}_{gt}(\tau_{U})$. This estimator is uniformly consistent in $\tau_{U}$ over a compact subset of $(0,1)$ as $n_{g}\to\infty$. We then form $\hat{Q}_{Y_{I0}(\tau_{U})}$, $\hat{Q}_{Y_{N1}(\tau_{U}^{*})}$ and $\hat{F}_{Y_{N0}(\tau_{U})}$ by the sample quantile functions and empirical CDF using variation in $\hat{Y}_{gt}(\tau_{U})$ across $g$ for fixed $t=0,1$. When $G\to\infty$, these estimators are again, uniformly consistent. 

Finally, construct a grid in $(0,1)$ for the $\tau_{U}$s, denoted by $\{\tau_{U,m}\}_{m=1,\ldots,M_{U}}$ where $M_{U}$ is finite and does not need to go to infinity with the sample size. By assuming continuity of $Q_{Y_{N1}(\tau_{U}^{*})}(\cdot)$ and $F_{Y_{N0}(\tau_{U,m})}$ for all $\tau_{U,m}$, the following estimator consistently estimate $Y^{N}_{I1}(\tau_{U}^{*},\tau_{V}^{*})$:
\begin{equation*}\label{eq.did.est}
  \hat{Y}^{N}_{I1}(\tau_{U}^{*},\tau_{V}^{*})=\frac{1}{M_{U}}\sum_{m=1}^{M_{U}}\hat{Q}_{Y_{N1}(\tau_{U}^{*})}\left(\left(\hat{F}_{Y_{N0}(\tau_{U,m})}\left(\hat{Q}_{Y_{I0}(\tau_{U,m})}\left(\tau_{V}^{*}\right)\right)\right)\right).
\end{equation*}

\subsection{Synthetic Control}
Again, suppose we have a data set $\{Y_{igt}:i=1,\ldots,n_{g};g=1,\ldots,G;t=0,1\}$ where $Y_{igt}$ are independently distributed across $i$ for each $g$, and can have mild serial correlation such that the sample quantile functions are consistent. For each $g,t$ and $\tau_{U}$, we construct $\hat{Y}_{gt}(\tau_{U})$ as before by requiring $n_{g}\to\infty$. Then for each $g$, we construct $\hat{Q}_{Y_{g,pre}(\tau_{U})}$ and $\hat{Q}_{Y_{g,post}(\tau_{U})}$ by sample quantiles of $\hat{Y}_{gt}(\tau_{U})$ over $t=1,\ldots,T_{0}$ and $t=T_{0}+1,\ldots,T$, respectively. 

Now, build grids $\{\tau_{V,m}\}_{m=1,\ldots,M_{V}}$. We estimate $\lambda^{*}(\tau_{U})$ by 
\begin{equation*}
  \hat{\lambda}^{*}(\tau_{U})=\arg\min_{\lambda(\tau_{U})\in\Delta^{G-1}}\frac{1}{M_{V}}\sum_{m=1}^{M_{V}}\left(\hat{Q}_{Y_{1,pre}(\tau_{U})}(\tau_{V,m})-\sum_{g=2}^{G}\lambda_{g}(\tau_{U})\hat{Q}_{Y_{g,pre}(\tau_{U})}(\tau_{V,m})\right)^{2}.
\end{equation*}
Note that $\hat{\lambda}^{*}(\tau_{U})$ is the solution to a least-square problem so is easy to compute. Consistency holds when $M_{V}\to\infty$, $T_{0}\to\infty$ and $T-T_{0}\to\infty$.

Finally, estimate $Y^{N}_{1,post}(\tau_{U},\tau_{V})$ by $\sum_{g=2}^{G}\hat{\lambda}^{*}_{g}(\tau_{U})\hat{Q}_{Y_{g,post}(\tau_{U})}(\tau_{V})$.

\begin{appendix}
\section*{Appendix}\label{appn} 
\setcounter{equation}{0}\renewcommand\theequation{A.\arabic{equation}} 
\begin{proof}[Proof of Theorem \ref{thm.did.id}]
Let $k=I$ if $g\in\mathcal{G}^{I}$ and $k=N$ if $g\in\mathcal{G}^{N}$. For any $\tau_{U}\in (0,1)$,
\begin{align*}
\Pr\left(Y_{igt}^{N}\leq h(Q_{U_{k}}(\tau_{U}),V_{gt},t)|V_{gt}\right)\notag=&\Pr\left(h(U_{igt},V_{gt},t)\leq h(Q_{U_{k}}(\tau_{U}),V_{gt},t)|V_{gt}\right)\notag\\
=&\Pr\left(U_{igt}\leq Q_{U_{k}}(\tau_{U})|V_{gt}\right)=\tau_{U},
\end{align*}
where the first equality is by model \eqref{eq.model} and Assumption \ref{assum.did.iid}, the second equality is by Assumption \ref{assum.did.mono}, and the last equality is by Assumption \ref{assum.did.ind}. Hence, $h(Q_{U_{k}}(\tau_{U}),V_{gt},t)$ is the $\tau_{U}$-th quantile of $Y^{N}_{igt}$ conditional on $V_{gt}$. Note that this quantity is equal to the $\tau_{U}$-th quantile of $Y^{N}_{igt}$ in the group $g$ and period $t$, denoted by $Y^{N}_{gt}(\tau_{U})$; it is a random variable whose randomness only comes from $V_{gt}$.

Using a similar argument, we further have
\begin{align}
\Pr\left(Y_{gt}^{N}(\tau_{U})\leq h(Q_{U_{k}}(\tau_{U}),Q_{V_{k}}(\tau_{V}),t)\right)=\Pr (V_{gt}\leq Q_{V_{k}}(\tau_{V}))=\tau_{V},\label{eq.did.doubleinvert}
\end{align}
where the first equality is by Assumption \ref{assum.did.mono}. 

\textit{Case (i)}. Under Assumption \ref{assum.did.iid} and the condition in case (i), let $Q_{U_{I}}=Q_{U_{N}}\equiv Q_{U}$. By the definition of $Y^{N}_{kt}(\tau_{U},\tau_{V})$ ($k=I,N$) and equation \eqref{eq.did.doubleinvert}, we have the following for all $t\in\{0,1\}$:
\begin{equation*}\label{eq.did.YNA}
Y^{N}_{kt}(\tau_{U},\tau_{V})=h(Q_{U}(\tau_{U}),Q_{V_{k}}(\tau_{V}),t),k=I,N.
\end{equation*}
Substituting equation \eqref{eq.potential}, we thus have
\begin{align}
 Q_{Y_{Nt}(\tau_{U})}(\tau_{V})=&h(Q_{U}(\tau_{U}),Q_{V_{N}}(\tau_{V}),t), t=0,1,\label{eq.did.YNB}\\
Q_{Y_{I0}(\tau_{U})}(\tau_{V})=&h(Q_{U}(\tau_{U}),Q_{V_{I}}(\tau_{V}),0). \label{eq.did.YNB2}
\end{align}

By Assumption \ref{assum.did.support}, there exists a $\tau_{V}'$ for the fixed $\tau_{V}^{*}$ such that $Q_{V_{I}}(\tau_{V}^{*})=Q_{V_{N}}(\tau_{V}')$. By monotonicity of $h$ in its second argument and by equations \eqref{eq.did.YNB} and \eqref{eq.did.YNB2}, $\tau_{V}'$ satisfies $Q_{Y_{I0}(\tau_{U})}(\tau_{V}^{*})=Q_{Y_{N0}(\tau_{U})}(\tau_{V}')$ for any $\tau_{U}$, or equivalently, $\tau_{V}'=F_{Y_{N0}(\tau_{U})}(Q_{Y_{I0}(\tau_{U})}(\tau_{V}^{*}))$. Therefore, for all $\tau_{U}\in (0,1)$,
\begin{align*}
  Y_{I1}^{N}(\tau_{U}^{*},&\tau_{V}^{*})=h(Q_{U}(\tau_{U}^{*}),Q_{V_{I}}(\tau_{V}^{*}),1)=h(Q_{U}(\tau_{U}^{*}),Q_{V_{N}}(\tau_{V}'),1)\\
  =&Q_{Y_{N1}(\tau_{U}^{*})}(\tau_{V}')=Q_{Y_{N1}(\tau_{U}^{*})}(F_{Y_{N0}(\tau_{U})}(Q_{Y_{I0}(\tau_{U})}(\tau_{V}^{*}))).
\end{align*}

\textit{Case (ii)}. Let $Q_{V_{I}}=Q_{V_{N}}\equiv Q_{V}$. Then similar to equations \eqref{eq.did.YNB} and \eqref{eq.did.YNB2}, we have
\begin{align}
 Q_{Y_{Nt}(\tau_{U})}(\tau_{V})=&h(Q_{U_{N}}(\tau_{U}),Q_{V}(\tau_{V}),t), t=0,1,\label{eq.did.YNB'}\\
Q_{Y_{I0}(\tau_{U})}(\tau_{V})=&h(Q_{U_{I}}(\tau_{U}),Q_{V}(\tau_{V}),0).\label{eq.did.YNB'2}
\end{align}

By Assumption \ref{assum.did.support}, there exists $\tau_{U}'$ for the fixed $\tau_{U}^{*}$ such that $Q_{U_{I}}(\tau_{U}^{*})=Q_{U_{N}}(\tau_{U}')$. By monotonicity of $h$ in its first argument and by equations \eqref{eq.did.YNB'} and \eqref{eq.did.YNB'2}, $\tau_{U}'$ can be found by $Q_{Y_{I0}(\tau_{U}^{*})}(\tau_{V})=Q_{Y_{N0}(\tau_{U}')}(\tau_{V})$ for any $\tau_{V}$. Therefore,
\begin{align*}
  Y_{I1}^{N}(\tau_{U}^{*},\tau_{V}^{*})=h(Q_{U_{I}}(\tau_{U}^{*}),Q_{V}(\tau_{V}^{*}),1)=h(Q_{U_{N}}(\tau_{U}'),Q_{V}(\tau_{V}^{*}),1)=Q_{Y_{N1}(\tau_{U}')}(\tau_{V}^{*}).
\end{align*}

\textit{Case (iii)}. In view of the monotonicity of $h$ and $\gamma$, we have the following equations:
\begin{equation*}\label{eq.did.YNB''}
\begin{aligned}
 &Q_{Y_{Nt}(\tau_{U})}(\tau_{V})=h(\gamma(Q_{U_{N}}(\tau_{U}),Q_{V_{N}}(\tau_{V})),t), t=0,1,\\
 &Q_{Y_{I0}(\tau_{U})}(\tau_{V})=h(\gamma(Q_{U_{I}}(\tau_{U}),Q_{V_{I}}(\tau_{V})),0).
 \end{aligned}
\end{equation*}
Therefore, by strict monotonicity of $h(\cdot,t)$, $Q_{Y_{I0}(\tau_{U}^{*})}(\tau_{V}^{*})=Q_{Y_{N0}(\tau_{U}')}(\tau_{V}')$ if and only if $\gamma(Q_{U_{N}}(\tau_{U}'), \allowbreak Q_{V_{N}}(\tau_{V}'))=\gamma(Q_{U_{I}}(\tau_{U}^{*}),Q_{V_{I}}(\tau_{V}^{*}))$; existence of such $(\tau_{U}',\tau_{V}')$ is guaranteed by Assumption \ref{assum.did.support}. The desired result obtains since $\gamma$ does not change in $t$.
\end{proof}
\begin{proof}[Proof of Theorem \ref{thm.did.test}]
We first prove part (i). The ``only if'' part follows Theorem \ref{thm.did.id} (i). Now we show the ``if'' part. By Assumptions \ref{assum.did.mono}-\ref{assum.did.iid} and equation \eqref{eq.potential}, we have
\begin{equation}\label{eq.did.implication}
Q_{Y_{k0}(\tau_{U})}(\tau_{V})=h(Q_{U_{k}}(\tau_{U}),Q_{V_{k}}(\tau_{V}),0),k=I,N.
\end{equation}
The condition $Q_{Y_{N0}(\tau_{U})}(\tau_{V}')=Q_{Y_{I0}(\tau_{U})}(\tau_{V})$ and equation \eqref{eq.did.implication} imply that 
\begin{equation}\label{eq.did.implication2}
    h(Q_{U_{N}}(\tau_{U}),Q_{V_{N}}(\tau_{V}'),0)=h(Q_{U_{I}}(\tau_{U}),Q_{V_{I}}(\tau_{V}),0)
\end{equation}
for all $\tau_{U}\in[0,1]$ because $\tau_{V}'$ does not depend on $\tau_{U}$. We now prove $Q_{V_{I}}(\tau_{V})=Q_{V_{N}}(\tau_{V}')$. Suppose $Q_{V_{I}}(\tau_{V})>Q_{V_{N}}(\tau_{V}')$. Since $h$ is strictly increasing in the first two arguments by Assumption \ref{assum.did.mono}, equation \eqref{eq.did.implication2} implies that $Q_{U_{I}}(\tau_{U})<Q_{U_{N}}(\tau_{U})$ for all $u\in [0,1]$. However, by Assumptions \ref{assum.did.support} and \ref{assum.did.compact}, there must exist a $\tau_{U}^{0}\in [0,1]$ such that $Q_{U_{I}}(\tau_{U}^{0})=Q_{U_{N}}(\tau_{U}^{0})$, a contradiction. We can similarly rule out the case of $Q_{V_{I}}(\tau_{V})<Q_{V_{N}}(\tau_{V}')$. Therefore, $Q_{V_{I}}(\tau_{V})=Q_{V_{N}}(\tau_{V}')$. Equation \eqref{eq.did.implication2} thus implies $h(Q_{U_{I}}(\tau_{U}),Q_{V_{I}}(\tau_{V}),0)=h(Q_{U_{N}}(\tau_{U}),Q_{V_{I}}(\tau_{V}),0)$ for all $\tau_{U}\in [0,1]$. It then has to be the case that $Q_{U_{I}}(\cdot)=Q_{U_{N}}(\cdot)$ on $[0,1]$ since $h$ is strictly increasing in the first argument.

The proof of part (ii) follows a very similar argument, so is omitted.

Part (iii). The ``if'' part. For any fixed $(\tau_{U}^{*},\tau_{V}^{*})\in [0,1]^{2}$, Assumption \ref{assum.did.support} implies that there exist $(\tau_{U},\tau_{V})\in [0,1]^{2}$ such that $Q_{U_{I}}(\tau_{U}^{*})=Q_{U_{N}}(\tau_{U})$ and $Q_{V_{I}}(\tau_{U}^{*})=Q_{V_{N}}(\tau_{U})$. So 
\begin{equation}\label{eq.test.3}
Q_{Y^{N}_{It}(\tau_{U}^{*})}(\tau_{V}^{*})=Q_{Y_{Nt}(\tau_{U})}(\tau_{V}),t=0,1.
\end{equation}
Then for any $(\tau_{U}',\tau_{V}')$ such that $Q_{Y_{N0}(\tau_{U})}(\tau_{V})=Q_{Y_{N0}(\tau_{U}')}(\tau_{V}')$, the condition in Theorem \ref{thm.did.test}-(iii) implies that $Q_{Y_{N1}(\tau_{U})}(\tau_{V})=Q_{Y_{N1}(\tau_{U}')}(\tau_{V}')$. Hence, equation \eqref{eq.test.3} implies that  $Y^{N}_{I1}(\tau_{U}^{*},\tau_{V}^{*})\equiv Q_{Y^{N}_{I1}(\tau_{U}^{*})}(\tau_{V}^{*})=Q_{Y_{N1}(\tau_{U}')}(\tau_{V}')$.

The ``only if'' part under $S(U_{I})=S(U_{N})$ and $S(V_{I})=S(V_{N})$. Suppose not. Then there exist some $(\tau_{U},\tau_{V})\in [0,1]^{2}$ and $(\tau_{U}',\tau_{V}')\in [0,1]^{2}$ such that $Q_{Y_{N0}(\tau_{U})}(\tau_{V})=Q_{Y_{N0}(\tau_{U}')}(\tau_{V}')$ but $Q_{Y_{N1}(\tau_{U})}(\tau_{V})\neq Q_{Y_{N1}(\tau_{U}')}(\tau_{V}')$. By $S(U_{I})=S(U_{N})$ and $S(V_{I})=S(V_{N})$, there exist $(\tau_{U}^{*},\tau_{V}^{*})\in [0,1]^{2}$ such that $Q_{U_{I}}(\tau_{U}^{*})=Q_{U_{N}}(\tau_{U})$ and $Q_{V_{I}}(\tau_{V}^{*})=Q_{V_{N}}(\tau_{V})$. Then we have 
$$Q_{Y_{I0}(\tau_{U}^{*})}(\tau_{V}^{*})=Q_{Y_{N0}(\tau_{U})}(\tau_{V})=Q_{Y_{N0}(\tau_{U}')}(\tau_{V}'),$$
but on the other hand, $Y^{N}_{I1}(\tau_{U}^{*},\tau_{V}^{*})\equiv Q_{Y^{N}_{I1}(\tau_{U}^{*})}(\tau_{V}^{*})=Q_{Y_{N1}(\tau_{U})}(\tau_{V})\neq Q_{Y_{N1}(\tau_{U}')}(\tau_{V}').$
We have thus found a pair $(\tau_{U}^{*},\tau_{V}^{*})\in [0,1]^{2}$ such that Theorem \ref{thm.did.id}-(iii) does not hold, a contradiction.
\end{proof}
\begin{proof}[Proof of Theorem \ref{thm.sc.id}]
Under Assumptions \ref{assum.did.ind}, \ref{assum.sc.mono}, \ref{assum.sc.iid} and equation \eqref{eq.potential}, we have $Y_{gt}(\tau_{U})=\tilde{h}(V_{gt},pre,\tau_{U})$ if $t\leq T_{0}$ for all $g$, and $Y_{gt}(\tau_{U})=\tilde{h}(V_{gt},post,\tau_{U})$ of $t>T_{0}$ and $g>1$. By isometry of $\tilde{h}$ in $j\in\{pre,post\}$ and by the requirement on the maps $V_{gt}\mapsto V_{gt'}$, the weights obtained at $j=pre$ also holds for $j=post$ following \cite{gunsilius2023distributional}.
\end{proof}
\end{appendix}
\newpage

\bibliographystyle{chicago}
\bibliography{references.bib} 
\end{document}